\documentclass[letterpaper, 10 pt, conference]{ieeeconf}  

\IEEEoverridecommandlockouts                              

\usepackage{amsmath} 
\usepackage{amssymb}  

\DeclareMathOperator{\tr}{Tr}
\DeclareMathOperator{\rank}{rank}
\DeclareMathOperator{\diag}{diag}

\usepackage{graphicx}
\newtheorem{lem}{Lemma}
\newtheorem{prop}{Proposition}
\newtheorem{theorem}{Theorem}

\newtheorem{alg}{Algorithm}
\newtheorem{rem}{Remark}

\title{\LARGE \bf Input Design for Model Discrimination and Fault Detection \\ via Convex Relaxation}

\author{Seunggyun Cheong$^{1}$ and Ian R. Manchester$^{1}$
\thanks{This work was supported by the Australian Research Council.}
\thanks{$^{1}$The authors are with Australian Centre for Field Robotics (ACFR) and School of the Aerospace, Mechanical and Mechatronic
Engineering, The University of Sydney, NSW 2006, Australia
        {\tt\small \{s.cheong,i.manchester\}@acfr.usyd.edu.au}}%
}

\begin{document}

\maketitle
\thispagestyle{empty}
\pagestyle{empty}

\begin{abstract}
This paper addresses the design  of input signals for the purpose of discriminating among a finite set of models dynamic systems within a given finite time interval. A motivating application is fault detection and isolation. We propose several specific optimization problems, with objectives or constraints based on signal power, signal amplitude, and probability of successful model discrimination. Since these optimization problems are nonconvex, we suggest a suboptimal solution via a random search algorithm guided by the semidefinite relaxation (SDR) and analyze the accuracy of the suboptimal solution. We conclude with a simple example taken from a benchmark problem on fault detection for wind turbines.
\end{abstract}

\section{INTRODUCTION}

In many applications of control and automation there will occur events that necessitate re-identifying a system model or detecting that the dynamics have changed. For example, it may be that the system dynamics have slowly changed due to aging, or abruptly changed to a fault. It is desirable to adjust the current control law accordingly or fix the fault. In this paper we investigate design of ``probing signals'' that improve the reliability of such a process. In particular, we consider the problem of discriminating among a fixed finite set of models, and finding the one which best matches the current system behaviour.

There is a long history of research into input design for system identification. The majority of papers have considered continuous parameterizations of models, with identification quality measured by estimated parameter variances or the Fisher information matrix. Basic approaches are summarised in \cite{Goodwin:1977}, \cite{Ljung:1999}, including pseudo-random binary signals and optimized multi-sine signals. In  \cite{Jansson:2005} semidefinite programming was used to design optimal signals in the frequency domain subject power constraints. Robust procedures were proposed in \cite{Rojas:2007}. In \cite{Manchester:2010}, \cite{Manchester:2012}, time domain signals were designed subject to power and amplitude constraints using semidefinite relaxation. 

Discriminating among a finite set of models is obviously more limiting in some ways than a continuous parameterization, but it also offers certain advantages, e.g. one can easily have different models of different order and structure. In the case of fault detection and isolation, there are frequently  a finite number modes of operation, corresponding to failures of different components, each of which is well understood: see, e.g., \cite{Hwang:2010} and references therein.

The majority of papers on fault detection assume the input signal is ``given'', and cannot be adjusted, and focus on the statistical estimation \cite{Hwang:2010}. In this paper we consider the case when the input can be adjusted to some small degree. The input design problem then could be posed in two general formats: maximize model discrimination probability subject to signal constraints related to nominal system operation, or minimize probing signal magnitude (in some sense) subject to constraints on reliability of model discrimination. These are related to the ``traditional'' and ``least costly'' input design, respectively, for continuously parameterized  model sets \cite{Rojas:2008}.

In \cite{Kerestecioglu:1994}, a frequency-domain approach to input design problem was proposed for model discrimination in terms of cumulative sum and probability ratio tests. In \cite{Campbell:2002} it was assumed that initial conditions of the system and disturbance signals are bounded by a known value, and the objective is an input signal with the least power such that it is impossible for models to have the same output signal. Although this method brings an absolute discrimination, the corresponding optimization problem is quite demanding. In contrast,  \cite{Skanda:2010} proposes to maximizing the Kullback-Leibler (K-L) divergence \cite{Kullback:1951} of the probability density functions (PDFs), corresponding to the output signals of models, from each other, assuming that initial conditions of the system can be chosen and a measurement noise signal is an iid random process with a normal distribution.


In the time domain, input-design problems have the structure of a nonconvex quadratic program. Recently, semidefinite relaxation techniques have been applied successfully to such problems in a wide variety of application areas \cite{Luo:2007} and in some cases can be proven to be very accurate, e.g. \cite{Nesterov:1998}. In this paper, we utilize and extend some of these methods for the problem of model discrimination.

The main contributions of this paper are: Section \ref{sec:ModelDiscrimination}: a model-selecting criterion, based on a modified version of the prediction error method (PEM) \cite{Ljung:1999} which admits rigorous analysis in terms of hypothesis testing; Section \ref{sec:InputDesign}: a family of optimization problems for input design, subject to different discrimination criteria and signal constraints; Section \ref{sec:Computation}: an approximate solution method for these (nonconvex) optimization problems using semidefinite relaxation; Section \ref{sec:Rank}: some cases when the proposed method is optimal, and  \ref{sec:quality}: an analysis of the quality of the solutions when the method is sub-optimal. We conclude by presenting a simple example based on a benchmark problem in fault detection for a wind turbine. All the proofs of the main results are provided in APPENDIX.

We use the following notation conventions: $|\cdot|$ denotes Euclidean norm of a finite-dimensional vector, $|\cdot|_\infty$ is the $\infty$-norm of a vector, $\|\cdot\|$ is the induced norm of a matrix with respect to $|\cdot |$. The set of symmetric positive semidefinite $n\times n$ matrices is denoted by $S_+^n$, the operator $\diag(\cdot)$ selects the diagonal elements of a square matrix, the expectation operator is denoted by $E[\cdot]$, and  $\chi_{1-\alpha,d}^2$ is the critical value of the chi-squared distribution with $d$ degrees of freedom and significance level $\alpha$.

\section{MODEL DISCRIMINATION}\label{sec:ModelDiscrimination}

We consider an uncertain system of the form
\begin{equation}\label{eq:sys}
y(t)=\mathcal{G}_0u(t)+\mathcal{H}_0s(t)+\Xi_0 x_0
\end{equation}
with single-input single-output (SISO) operators $\mathcal{G}_0$  and $\mathcal{H}_0$, and where $u$ is an input signal that we apply to the system, $y$ is an observed output signal, and $s$ is an unobserved signal that produces the additive disturbance, and $x_0$ is an initial condition (if present). We suppose that, based on a priori knowledge of the system in \eqref{eq:sys}, there is available a finite number $N$ of causal, linear time-invariant (LTI), discrete-time, and SISO models
\begin{equation}\label{eq:Models}
y_n(t)=\mathcal{G}_nu_n(t)+\mathcal{H}_ns_n(t)+\Xi_n x_n,\ t=0,1,\cdots\\
\end{equation}
for $n=1,\cdots,N$ each of which satisfies the following:
\begin{enumerate}
\item The initial conditions of each model $\mathcal{G}_n$  is described by a finite-dimensional vector $x_n=\bar x_n+Q_nv_n$ where $\bar x_n$ and $Q_n$ are predetermined vector and matrix, respectively, and $v_n$ is a Gaussian random vector that has a zero mean and a covariance matrix $\sigma_n^2I$ for some constant $\sigma_n$.
\item The disturbance dynamics $\mathcal{H}_n$ is invertible and monic and has a zero initial condition at time $t=0$.
\item The signal $s_n(t), t=0,1,\cdots$ is an iid random process and is also independent of $x_n$. Each $s_n(t)$ has a normal distribution, denoted by $f_{\sigma_n}$, with a zero mean value and some variance $\sigma_n^2$.
\end{enumerate}
Note that the appearance of $\sigma_n$ in the first and third assumption is just for notational simplicity, and does not involve any loss of generality because of the flexibility in $Q_n$.

Naturally, we assume that the models are distinct.  Furthermore, the models may have different orders, which means that the dimensions of $x_n$'s may be different. For simplicity, the elements of $v_n$ and $s_n(t)$ have the same variance $\sigma_n^2$, which works as a fitting parameter later.

Let $T$ be a positive integer and represent the length of an experiment. Then, using \eqref{eq:Models}, we can describe the output signal of the $n$-th model in ``lifted form'': $\mathbf{y}_n\triangleq\begin{bmatrix}y_n(0) &\cdots &y_n(T-1)\end{bmatrix}'$ by
\begin{equation}\label{eq:yn}
\begin{split}
\mathbf{y}_n&=G_n\mathbf{u}_n+\Psi_nx_n+H_n\mathbf{s}_n\\
&=G_n\mathbf{u}_n+\Psi_n\bar x_n+\Psi_nQ_nv_n+H_n\mathbf{s}_n
\end{split}
\end{equation}
with $\mathbf{u}_n=\begin{bmatrix}u_n(0) &\cdots &u_n(T-1)\end{bmatrix}'$ and $\mathbf{s}_n=\begin{bmatrix}s_n(0) &\cdots &s_n(T-1)\end{bmatrix}'$ where the matrices $G_n$, $\Psi_n$, and $H_n$ represent the $n$-th model. For example, if $(A_n, B_n, C_n, D_n)$ is a state-space representation of the $n$-th model, then we have
\begin{equation}\nonumber
G_n=\begin{bmatrix}g_{n,0} &0 &\cdots &0\\ g_{n,1} &g_{n,0} &\cdots &0\\ \vdots &\vdots &\ddots &\vdots\\ g_{n,T-1} &g_{n,T-2} &\cdots &g_{i,0}\end{bmatrix}, \Psi_n=\begin{bmatrix}C_n\\ C_nA_n\\ \vdots\\ C_nA_n^{T-1}\end{bmatrix}
\end{equation}
with $g_{n,0}=D_n$ and $g_{n,i}=C_nA_n^{i-1}B_n$ for $i=1,2,\cdots$. The matrices $H_n$'s are defined similarly to $G_n$'s. Note that the matrices $G_n$ and $H_n$ are lower triangular due to the causality of the models and, in particular, the  $H_n$'s are invertible.

For the purposes of model discrimination, given input-output data of the real system: $\mathbf{u}=\begin{bmatrix}u(0) &\cdots &u(T-1)\end{bmatrix}'$ and $\mathbf{y}=\begin{bmatrix}y(0) &\cdots &y(T-1)\end{bmatrix}'$,  we construct for each model a vector $\tilde v_n$ and a signal $\tilde s_n(t)$ like so:
\begin{equation}\label{eq:fic}
\tilde p_n\triangleq\begin{bmatrix}\tilde v_n\\ \mathbf{\tilde s}_n\end{bmatrix}=\begin{bmatrix}\Psi_nQ_n &H_n\end{bmatrix}^+\left(\mathbf{y}-G_n\mathbf{u}-\Psi_n\bar x_n\right)
\end{equation}
where $\mathbf{\tilde s}_n\triangleq\begin{bmatrix}\tilde s_n(0) &\cdots &\tilde s_n(T-1)\end{bmatrix}'$ and $+$ means the Moore-Penrose pseudo inverse. These signals represent the initial condition and noise signals that would have been necessary for system $n$ to have produced the observed input-output data set, and parallels the use of ``fictitious'' reference signals in unfalsified adaptive control (See, for example, \cite{Safonov:1997} and \cite{Cheong:2010} for details). 

If the elements of $\tilde p_n$ in \eqref{eq:fic} are plausible realisations of the random process associated with the $n$-th model, then we may conclude that the $n$-th model successfully describes the input-output data $u(t)$ and $y(t)$. Note that if the initial condition $x_n$ is a deterministic vector, i.e. $Q_n=\mathbf{0}$, we remove $\tilde v_n$ and $\Psi_nQ_n$ in \eqref{eq:fic} and the signal $\tilde s_n(t)$ reduces to the one-step-ahead prediction error \cite{Ljung:1999} corresponding to the $n$-th model. The presence of initial conditions for our criterion is important in fault detection because we expect to be examining the system during its normal operation, and the statistics of the initial conditions may come from, e.g., a bank of Kalman filters \cite{Hwang:2010}.

Based on the signals $\tilde p_n$, we can use the maximum likelihood method to select a particular model. Denote the dimension of $\tilde p_n$ by $T_n$ so that we have $T_n=T+\dim v_n$. Also, denote by $\mathbf{f}_\sigma$ the PDF of a $T_n$-dimensional multivariate normal distribution with a zero mean vector and a covariance matrix $\sigma^2I$. Given $\tilde p_n=\begin{bmatrix}\tilde v_n{}' &\mathbf{\tilde s}_n{}'\end{bmatrix}'$ in \eqref{eq:fic}, an estimate
\begin{equation}\label{eq:SigmaHatn}
\begin{split}
\hat\sigma_n&=\arg\max_{\sigma\ge 0}\ln\mathbf{f}_\sigma(\tilde p_n)\\
&=\arg\min_{\sigma\ge 0}T_n\ln\sigma^2+\frac{1}{\sigma^2}\tilde p_n{}'\tilde p_n\\
&=\left(\frac{1}{T_n}\tilde p_n{}'\tilde p_n\right)^\frac{1}{2}=\left\{\frac{1}{T_n}\left(\tilde v_n{}'\tilde v_n+\mathbf{\tilde s}_n{}'\mathbf{\tilde s}_n\right)\right\}^\frac{1}{2}.
\end{split}
\end{equation}
provides the maximum likelihood (over $\sigma$) for the $n$-th model, and then we select a model that has the least value for $\hat\sigma_n^2$, i.e.
\begin{equation}\label{eq:nhat}
\begin{split}
\hat n&=\arg\min_{n\in\{1,\cdots,N\}}\hat\sigma_n^2\\
&=\arg\min_{n\in\{1,\cdots,N\}}\frac{1}{T_n}\left(\tilde v_n{}'\tilde v_n+\mathbf{\tilde s}_n{}'\mathbf{\tilde s}_n\right).
\end{split}
\end{equation}

The selection criterion in \eqref{eq:nhat} provides a definitive selection and only one model stands after this method is applied to collected data.
On the other hand, sometimes we may want to consider all the models that do not show some evident inadequacy. We formulate this as a hypothesis testing for each model with
\begin{quote}
Null Hypothesis : The $n$-th model produced the data with some $\sigma_n$ less than or equal to a known value $\bar\sigma\ge 0$.
\end{quote}
Unless this null hypothesis is rejected based on the collected data, we keep the $n$-th model as a potential origin of the data. Thus, possibly multiple models remain as candidates after this testing.

For the $n$-th model, we may reject the null hypothesis if the data shows a significant evidence that the estimated variance $\hat\sigma_n^2$ in \eqref{eq:SigmaHatn} is greater than $\bar\sigma^2$. Thus, using a chi-squared test, we reject the null hypothesis when a statistic $\frac{T_n\hat\sigma_n^2}{\bar\sigma^2}$ is greater than $\chi_{1-\alpha,T_n-1}^2$. Therefore, the candidate models based on the collected data are described by a set
\begin{equation}\label{eq:nset}
\left\{n\in\{1,\cdots,N\}\left|\hat\sigma_n^2\le\frac{\chi_{1-\alpha,T_n-1}^2}{T_n}\bar\sigma^2\right.\right\}.
\end{equation}

\section{INPUT DESIGN FOR MODEL DISCRIMINATION}\label{sec:InputDesign}

If the real system in \eqref{eq:sys} is well-described by one of the $N$ models, it is important to be able to distinguish it from other models. In this section, we consider design of a probing input signal $u$ that ensures the model selection procedure in Section \ref{sec:ModelDiscrimination} is successful.

Suppose that we wish to discriminate between models $n_1$ and $n_2$, when in reality the system in \eqref{eq:sys} is compatible with model $n_1$ with a noise variance bounded by a known constant $\sigma_{n_1}\le \bar\sigma$. Then, using \eqref{eq:yn}, the output signal $\mathbf{y}=\begin{bmatrix}y(0) &\cdots &y(T-1)\end{bmatrix}'$ is given by
\begin{equation}\nonumber
\mathbf{y}=G_{n_1}\mathbf{u}_{n_1}+\Psi_{n_1}\bar x_{n_1}+\Psi_{n_1}Q_{n_1}v_{n_1}+H_{n_1}\mathbf{s}_{n_1}.
\end{equation}
Now, if we test the data against model $n_2$ by applying the procedure in \eqref{eq:fic}, we find that $\tilde p_{n_2}=\begin{bmatrix}\tilde v_{n_2}{}' &\mathbf{\tilde s}_{n_2}{}'\end{bmatrix}'$ is
\begin{equation}\label{eq:ptilden2}
\tilde p_{n_2}=\tilde\mu_{n_2n_1}+\tilde\Sigma_{n_2n_1}p_{n_1}
\end{equation}
where $p_{n_1}=\begin{bmatrix}v_{n_1}{}' &\mathbf{s}_{n_1}{}'\end{bmatrix}'$ and
\begin{equation}\nonumber
\begin{split}
\tilde\mu_{n_2n_1}&=\begin{bmatrix}\Psi_{n_2}Q_{n_2} &H_{n_2}\end{bmatrix}^+\left(G_{n_1}-G_{n_2}\right)\mathbf{u}+\eta_{n_2n_1}\\
\eta_{n_2n_1}&=\begin{bmatrix}\Psi_{n_2}Q_{n_2} &H_{n_2}\end{bmatrix}^+\left(\Psi_{n_1}\bar x_{n_1}-\Psi_{n_2}\bar x_{n_2}\right)\\
\tilde\Sigma_{n_2n_1}&=\begin{bmatrix}\Psi_{n_2}Q_{n_2} &H_{n_2}\end{bmatrix}^+\begin{bmatrix}\Psi_{n_1}Q_{n_1} &H_{n_1}\end{bmatrix}.
\end{split}
\end{equation}
Similarly,
\begin{equation}\label{eq:ptilden1n2}
\begin{split}
\tilde p_{n_1}&=\begin{bmatrix}\Psi_{n_1}Q_{n_1} &H_{n_1}\end{bmatrix}^+\begin{bmatrix}\Psi_{n_1}Q_{n_1} &H_{n_1}\end{bmatrix}p_{n_1}\\
&=\tilde\mu_{n_1n_2}+\tilde\Sigma_{n_1n_2}\tilde p_{n_2}.
\end{split}
\end{equation}
Since $p_{n_1}$ is a Gaussian random vector with a zero mean vector and a covariance $\sigma_{n_1}^2I$, it follows, from \eqref{eq:ptilden2}, that $\tilde p_{n_2}$ is a normally distributed random vector with a mean vector $\tilde\mu_{n_2n_1}$ and a covariance matrix $\sigma_{n_1}^2\tilde\Sigma_{n_2n_1}\tilde\Sigma_{n_2n_1}{}'$. Note in particular that $\tilde\mu_{n_2n_1}$ is an affine function of $\mathbf{u}$ and all other quantities are independent of $\mathbf{u}$. In the following theorem it is shown that if $\left |\tilde\mu_{n_2n_1}\right |$ is sufficiently large for any $n_2\ne n_1$, then the hypothesis testing in Section \ref{sec:ModelDiscrimination} brings statistically reliable results.

\ 

\begin{theorem}\label{prop:ID}
Suppose that real system data are compatible with the $n_1$-th model with some $\sigma_{n_1}$ less than or equal to a known value $\bar\sigma$. If, for any $n_2\in\left\{1,\cdots,N\right\}\setminus\{n_1\}$, either
\begin{equation}\label{eq:MuTildeCondition1}
|\tilde\mu_{n_2n_1}|>\left(\chi_{1-\alpha,T_{n_2}-1}+\chi_{1-\alpha,T_{n_1}-1}\left\|\tilde\Sigma_{n_2n_1}\right\|\right)\bar\sigma
\end{equation}
or
\begin{equation}\label{eq:MuTildeCondition2}
|\tilde\mu_{n_1n_2}|>\left(\chi_{1-\alpha,T_{n_1}-1}+\chi_{1-\alpha,T_{n_2}-1}\left\|\tilde\Sigma_{n_1n_2}\right\|\right)\bar\sigma
\end{equation}
then, with at least $100\times (1-\alpha)\%$ probability, only the $n_1$-th model is not rejected by the hypothesis test \eqref{eq:nset}.
\end{theorem}

\ 

Note that the probability of the correct model being selected increases as we increase the critical values of the chi-squared distribution in \eqref{eq:MuTildeCondition1} and \eqref{eq:MuTildeCondition2}.

Based on this proposition, we design an input signal satisfying 
\begin{equation}\label{eq:gamma}
|\tilde\mu_{n_2n_1}|>\gamma_{(n_1,n_2)}
\end{equation}
for all $n_1,n_2\in\{1,\cdots,N\}$ satisfying $n_1<n_2$ where
\begin{equation}\nonumber
\begin{split}
\gamma_{(n_1,n_2)}&=\max\left\{\left(\chi_{1-\alpha,T_{n_2}-1}+\chi_{1-\alpha,T_{n_1}-1}\left\|\tilde\Sigma_{n_2n_1}\right\|\right)\bar\sigma,\right.\\
&\hspace{13mm}\left.\left(\chi_{1-\alpha,T_{n_1}-1}+\chi_{1-\alpha,T_{n_2}-1}\left\|\tilde\Sigma_{n_1n_2}\right\|\right)\bar\sigma\right\}.
\end{split}
\end{equation}
Note that, for any given unordered pair $(n_1,n_2)$ or any given two models, there is only one condition imposed by \eqref{eq:gamma}. Thus, the total number of conditions is $\frac{N(N-1)}{2}$.

Roughly speaking, purpose of the conditions \eqref{eq:MuTildeCondition1} and \eqref{eq:MuTildeCondition2} is to make $\tilde p_{n_1}$ and $\tilde p_{n_2}$ for any $n_2\ne n_1$ to be far apart from each other with high probability by placing the vectors $\tilde\mu_{n_2n_1}$ and $\tilde\mu_{n_1n_2}$ away from the zero vector. This can be achieved by suitably chosen $\mathbf{u}$. Alternatively, when the dimensions of $\tilde p_{n_1}$ and $\tilde p_{n_2}$ are the same, we can also pursue the same goal by increasing the Kullback-Leibler divergence
\begin{equation}\nonumber
\begin{split}
D_{KL}(\mathbf{\tilde f}_{n_2}||\mathbf{\tilde f}_{n_1})&=\frac{1}{2\sigma_{n_1}^2}\left|\tilde\mu_{n_2n_1}\right|^2+\frac{1}{2}\left\{{\rm{Tr}}\left(\tilde\Sigma_{n_2n_1}\tilde\Sigma_{n_2n_1}{}'\right)\right.\\
&\hspace{4mm}\left.-T_{n_2}-\ln\det\left(\tilde\Sigma_{n_2n_1}\tilde\Sigma_{n_2n_1}{}'\right)\right\}
\end{split}
\end{equation}
of $\mathbf{\tilde f}_{n_1}$ from $\mathbf{\tilde f}_{n_2}$ where $\mathbf{\tilde f}_{n_2}$ and $\mathbf{\tilde f}_{n_1}$ are the PDFs of $\tilde p_{n_2}$ and $\tilde p_{n_1}$, respectively. This formulation was studied in \cite{Skanda:2010}, however our condition in \eqref{eq:gamma}, combined with the hypothesis testing in Section \ref{sec:ModelDiscrimination}, provides an explicit reliability measure for model discrimination.

\subsection{A Family of Nonconvex Quadratic Optimization Problems}

We have shown that model discrimination is improved by  increasing the norm of the vectors $\tilde\mu_{n_2n_1}$ for each pair of models $n_1, n_2$. With this in mind, there are a number of reasonable formulations of specific optimization problems, trading off different measures of signal size and model discrimination ability.

Since each $\tilde\mu_{n_2n_1}$ is an affine functions of the control input $\mathbf{u}$, unless constraints are imposed on the input it is clear that maximizing model discrimination leads to an unbounded and ill-posed optimization problem. Given some positive bounds $\bar u, \bar y$ we consider some natural constraints:
\begin{align}
Z_{i2}(\mathbf{u}):=&\frac{1}{\bar u^2}|\mathbf{u}|^2,\label{eq:InputPower}\\
Z_{o2}(\mathbf{u}):= &\frac{1}{\bar y^2}\max_{n\in\{1,\cdots,N\}}|G_n\mathbf{u}+\Psi_n\bar x_n|^2,\label{eq:OutputPower}\\
Z_{i\infty}(\mathbf{u}):=&\frac{1}{\bar u}|\mathbf{u}|_\infty,\label{eq:InputAmplitude}\\
Z_{o\infty}(\mathbf{u}):= &\frac{1}{\bar y}\max_{n\in\{1,\cdots,N\}}|G_n\mathbf{u}+\Psi_n\bar x_n|_\infty,\label{eq:OutputAmplitude}
\end{align}
where the subscripts refer to input and output, and 2-norm and $\infty$-norm. Notice that the constraints on the output are the maximum over all models in the set. There may be other constraints that are natural to consider, e.g. move size (change in $u$), or the size of other states in a state-space model, or a weighted combination of inputs and outputs as in a linear quadratic regulator. These could also easily be used in the framework we propose.

Let $M=\frac{N(N-1)}{2}$ be the total number of the unordered pairs of the models, and
\begin{align}
\overline G_m &:=\begin{bmatrix}\Psi_{n_2}Q_{n_2} &H_{n_2}\end{bmatrix}^+\left(G_{n_1}-G_{n_2}\right),\\ 
\overline\eta_m& :=\eta_{n_2n_1}, \ \ \overline\gamma_m:=\gamma_{(n_1,n_2)}.
\end{align}
We consider two natural measures of model discrimination, a ``worst case'' criterion, derived from \eqref{eq:gamma}:
\begin{equation}\label{eq:Discrimination}
V_\infty(\mathbf{u}):=\min_{m\in\{1,\cdots,M\}}\frac{1}{\overline\gamma_m^2}\left|\overline G_m\mathbf{u}+\overline\eta_m \right|^2
\end{equation}
and a somewhat easier ``weighted average'' case:
\begin{equation}\label{eq:AverageDiscrimination}
V_2(\mathbf{u}):=\frac{1}{M}\sum_{m=1}^M\frac{w_m}{\overline\gamma_m^2}\left|\overline G_m\mathbf{u}+\overline\eta_m \right|^2\\
\end{equation}
where $w_m$'s are the weights. The latter may be appropriate if based on prior data certain models are highly likely and should be favored for discrimination.

Note that all of  \eqref{eq:InputPower}-\eqref{eq:OutputAmplitude} and \eqref{eq:Discrimination}, \eqref{eq:AverageDiscrimination} are convex quadratic functions of $\mathbf{u}$. With these signal properties and discrimination factors, we can consider either the ``least costly'' input signal for guaranteed discrimination: 
\begin{equation}\label{eq:InputModel}
\begin{split}
&\min_{\mathbf{u}\in\mathbb{R}^T} Z_a(\mathbf{u})\\
&\hspace{2mm}{\rm{s.t.}}\ V_b(\mathbf{u})\ge 1
\end{split}
\end{equation}
with $a\in\{i2, o2, i\infty, o\infty\}$ and $b\in\{\infty, 2\}$. Alternatively, we can consider maximizing discrimination reliability subject to hard constraints on the input:
\begin{equation}\label{eq:ModelInput}
\begin{split}
&\max_{\mathbf{u}\in\mathbb{R}^T}V_b(\mathbf{u})\\
&\hspace{1mm}{\rm{s.t.}}\ Z_a(\mathbf{u})\le 1.
\end{split}
\end{equation}
Both of these formulations are {\em nonconvex} quadratic optimization problems, due to the constraint in \eqref{eq:InputModel} and the maximization in \eqref{eq:ModelInput}. There is no known polynomial-time algorithm for nonconvex quadratic optimization, and it is generally considered unlikely one will be found since they belong to the class of NP-hard problems \cite{Luo:2007}. For this reason, we pursue a convex relaxation technique.

\section{RELAXATION TO A SEMIDEFINITE PROGRAM}\label{sec:Computation}

Semidefinite relaxation is a general approach for problems of the form $\min_x f_0(x)$ subject to $f_i(x) \le 0$ where $x\in\mathbb R^{n-1}$ and $f_i(x) = x'A_ix+2b_i'x+c$, $i = 0, 1, 2, ..., m$. In particular, it is not assumed that $A_i$ are positive semidefinite (which would bring convexity). These functions can be homogenised as
\[
f_i(x) = \begin{bmatrix}x\\1\end{bmatrix}' \begin{bmatrix}A_i & b_i' \\ b_i & c_i\end{bmatrix}\begin{bmatrix}x\\1\end{bmatrix} =: \xi' Q_i \xi.
\]
The simple fact that $\xi' Q_i \xi=\tr (\xi' Q_i \xi) = \tr (Q_i \xi\xi')$, and the fact that any $n\times n$ matrix $X=X'\ge 0$ which is rank one can be decomposed as $X=\xi \xi'$ leads to the equivalent problem:
\[
\min_{X\in S_+^n} \tr(Q_0 X):\tr(Q_iX)\le 0, X_{n,n} = 1, \rank(X) = 1.
\]
This is an exact reformulation of the problem, and all constraints are convex except for the rank constraint. Semidefinite relaxation consists of dropping the rank constraint, which results in a semidefinite program \cite{Nesterov:1998}, \cite{Luo:2007}. Relaxations always give an ``optimistic'' value, since they optimize the same objective function over a larger feasible set. The quality of a relaxation is determined by the gap between the true optimum and the optimum of the relaxation.

To apply this idea to the model discrimination problem, we construct homogenised forms of the key quantities from the previous section. The discrimination factor is given by
\begin{equation}\label{eq:homogenization1}
\left|\overline G_m\mathbf{u}+\overline\eta_m\right|^2=\begin{bmatrix}\mathbf{u}\\ 1\end{bmatrix}'\begin{bmatrix}\overline G_m{}'\\ \overline\eta_m{}'\end{bmatrix}\begin{bmatrix}\overline G_m &\overline\eta_m\end{bmatrix}\begin{bmatrix}\mathbf{u}\\ 1\end{bmatrix}
\end{equation}
and the various signal constraints can be represented as
\begin{equation}\label{eq:homogenization2}
\begin{split}
\left|\mathbf{u}\right|^2&=\begin{bmatrix}\mathbf{u}\\ 1\end{bmatrix}'\begin{bmatrix}\mathbf{I} &\mathbf{0}\\ \mathbf{0}' &0\end{bmatrix}\begin{bmatrix}\mathbf{u}\\ 1\end{bmatrix},\\
\left|G_n\mathbf{u}+\Psi_n\bar x_n\right|^2&=\begin{bmatrix}\mathbf{u}\\ 1\end{bmatrix}'\begin{bmatrix}G_n{}'\\ \bar x_n{}'\Psi_n{}'\end{bmatrix}\begin{bmatrix}G_n &\Psi_n\bar x_n\end{bmatrix}\begin{bmatrix}\mathbf{u}\\ 1\end{bmatrix},\\
\left|\mathbf{u}_i\right|^2&=\begin{bmatrix}\mathbf{u}\\ 1\end{bmatrix}'\begin{bmatrix}e_ie_i{}' &\mathbf{0}\\ \mathbf{0}' &0\end{bmatrix}\begin{bmatrix}\mathbf{u}\\ 1\end{bmatrix},\\
\left|\left(G_n\mathbf{u}+\Psi_n\bar x_n\right)_i\right|^2&=\begin{bmatrix}\mathbf{u}\\ 1\end{bmatrix}'\begin{bmatrix}G_n{}'\\ \bar x_n{}'\Psi_n{}'\end{bmatrix}e_ie_i{}'\begin{bmatrix}G_n &\Psi_n\bar x_n\end{bmatrix}\begin{bmatrix}\mathbf{u}\\ 1\end{bmatrix}.
\end{split}
\end{equation}
where $e_i$ is the indicator vector for element $i$, and $|x|_\infty\le a$ can be imposed for any vector by $|x_i|^2\le a^2 \ \forall \ i$.

These can be equivalently represented in terms of a matrix variable $U\in S_+^{T+1}$:
\begin{equation}\nonumber
\begin{split}
\left|\overline G_m\mathbf{u}+\overline\eta_m\right|^2&={\rm{Tr}}\left(\begin{bmatrix}\overline G_m{}'\\ \overline\eta_m{}'\end{bmatrix}\begin{bmatrix}\overline G_m &\overline\eta_m\end{bmatrix}U\right)\\
\left|\mathbf{u}\right|^2&={\rm{Tr}}\left(\begin{bmatrix}\mathbf{I} &\mathbf{0}\\ \mathbf{0}' &0\end{bmatrix}U\right)\\
\left|G_n\mathbf{u}+\Psi_n\bar x_n\right|^2&={\rm{Tr}}\left(\begin{bmatrix}G_n{}'\\ \bar x_n{}'\Psi_n{}'\end{bmatrix}\begin{bmatrix}G_n &\Psi_n\bar x_n\end{bmatrix}U\right)\\
\left|\mathbf{u}_i\right|^2&={\rm{Tr}}\left(\begin{bmatrix}e_ie_i{}' &\mathbf{0}\\ \mathbf{0}' &0\end{bmatrix}U\right)\\
\left|\left(G_n\mathbf{u}+\Psi_n\bar x_n\right)_i\right|^2&={\rm{Tr}}\left(\begin{bmatrix}G_n{}'\\ \bar x_n{}'\Psi_n{}'\end{bmatrix}e_ie_i{}'\begin{bmatrix}G_n &\Psi_n\bar x_n\end{bmatrix}U\right).
\end{split}
\end{equation}
with additional conditions $U_{T+1, T+1}=1$, $U\ge 0$, and ${\rm{rank}}(U)=1$. The optimization problems are now described in terms of $U$ instead of $\mathbf{u}$. The SDR is completed by dropping the rank constraint ${\rm{rank}}(U)=1$.

Given the various $Z_a(\mathbf u)$ and $V_b(\mathbf u)$ from the previous section, we denote the semidefinite relaxation forms of these by $\hat Z_a(U)$ and $\hat V_b(U)$. So we can again have general problems of the form
\begin{equation}\label{eq:InputModelSDP}
\min_{U\in S_+^{T+1}}\hat Z_a(U) {\textrm{ s.t. }}\ \hat V_{b}(U)\ge 1, U_{T+1, T+1} = 1.
\end{equation}
or
\begin{equation}\label{eq:ModelInputSDP}
\max_{U\in S_+^{T+1}}\hat V_b(U) {\textrm{ s.t. }}\ \hat Z_a(U)\le 1,U_{T+1, T+1} = 1.
\end{equation}
as relaxations of \eqref{eq:InputModel} and \eqref{eq:ModelInput}, respectively. Clearly, one can also add multiple constraints (e.g. on the input and output) and retain the same SDP structure, or alternative cost functions such as LQR. 
We do not give every possible permutation here, but for example, the relaxation of 
\begin{equation}
\begin{split}
&\max_{\mathbf{u}\in\mathbb{R}^T}V_2(\mathbf{u})\\
&\hspace{1mm}{\rm{s.t.}}\ Z_{i\infty}(\mathbf{u})\le 1,
\end{split}
\end{equation}
is given by
\begin{equation}
\begin{split}
&\max_{U\in S_+^{T+1}}\hat V_2(U)\\
&\hspace{5mm}{\rm{s.t.}}\ \hat Z_{i\infty}(U)\le 1,\\
&\hspace{10mm}U_{T+1, T+1} = 1.
\end{split}
\end{equation}
where
\begin{align}
\hat V_2(U)&=\frac{1}{m}\sum_{m=1}^M \frac{w_m}{\bar\gamma_m^2}{\rm{Tr}}\left(\begin{bmatrix}\overline G_m{}'\\ \overline\eta_m{}'\end{bmatrix}\begin{bmatrix}\overline G_m &\overline\eta_m\end{bmatrix}U\right),\\
\hat Z_\infty(U)&=\frac{1}{\bar u^2}\max_{i \in \{0, 1, ..., T\}}{\rm{Tr}}\left(\begin{bmatrix}e_ie_i{}' &\mathbf{0}\\ \mathbf{0}' &0\end{bmatrix}U\right).
\end{align}

The SDP relaxations generally give ``optimistic'' results, i.e. under-estimations of required signal power in \eqref{eq:InputModel} and over-estimations of discrimination power in \eqref{eq:ModelInput}. However, the advantage is that they can be efficiently solved (polynomial time to a given accuracy) using freely available solvers such as Sedumi \cite{Sturm:1999} and interfaces such as Yalmip \cite{lofberg2004yalmip} and CVX \cite{CVX:2012}.

\section{OPTIMAL SOLUTIONS OF THE TRUE PROBLEMS FROM THE RELAXATION}\label{sec:Rank}

Although in general the objective value of \eqref{eq:InputModel} and \eqref{eq:InputModelSDP}, or \eqref{eq:ModelInput} and \eqref{eq:ModelInputSDP}, will be different, there are some cases in which they are the same. That is, there is no gap between the optimal values of the relaxed problem and the true problem. In particular, this is the case for either the ``least costly'' formulation \eqref{eq:InputModel} or ``traditional'' formulation \eqref{eq:ModelInput} under the following situations:
\begin{enumerate}
\item The signal conditions $Z_a(\mathbf{u})$ are of ``power'' type, i.e. $a=i2$ or $o2$ or another quadratic, e.g. LQR, \textbf{and}
\begin{enumerate}
\item The model discrimination constraint is of the weighted average form $V_2(\mathbf{u})$, \textbf{or}
\item The model discrimination is absolute $V_\infty(\mathbf{u})$ but there is only two models to discriminate between, giving $m=1$ in \eqref{eq:Discrimination}.
\end{enumerate}
\end{enumerate}
Under any of these scenarios,  \eqref{eq:InputModel} or \eqref{eq:ModelInput} become problems to optimize (minimize or maximize) a quadratic function subject to a single quadratic constraint. For this problem structure, it is known that the semidefinite relaxation has no gap \cite{Polyak:1998}, \cite[Appendix B]{boyd2004convex}.

Let us also consider a more general case involving more constraints or models to distinguish:
\begin{equation}\label{eq:SDP}
\begin{split}
&\min_{U\in S_+^{T+1}}\max_{\ell\in\{1,\cdots,L\}}{\rm{Tr}}(P_\ell U)\\
&\hspace{8mm}{\rm{s.t.}}\ {\rm{Tr}}(R_kU)\ge r,\ k=1,\cdots,K\\
&\hspace{14mm}{\rm{Tr}}(e_{T+1}e_{T+1}{}'U)=1,
\end{split}
\end{equation}
where $L$ and $K$ are positive integers, $r$ is a constant, and $P_\ell$'s and $R_k$'s are $(T+1)\times (T+1)$ matrices. It should be clear that SDPs in \eqref{eq:InputModelSDP} and \eqref{eq:ModelInputSDP} are special cases of this SDP. In particular, the ``zero gap'' cases we described above correspond to $L=K=1$.

Then, the following proposition is a straightforward extension \cite{Pataki:1998}, \cite{Barvinok:1995}, \cite{YeLecture}. 
\begin{prop}\label{prop:RankReductionL1}
If the optimization problem in \eqref{eq:SDP} with $L=1$ has an optimal solution, then there exists an optimal solution $U^*$ satisfying
\begin{equation}\nonumber
\frac{{\rm{rank}}(U^*)\left({\rm{rank}}(U^*)+1\right)}{2}\le K+1.
\end{equation}
\end{prop}
\ 

\begin{rem}\label{rem:Rank1}
When $L=K=1$, as in the cases described above,  Proposition \ref{prop:RankReductionL1} guarantees that a solution $U^*$ to the corresponding SDP has rank $1$, which implies that there exists $\mathbf{u}^*$ satisfying $U^*=\begin{bmatrix}\mathbf{u}^*{}' &1\end{bmatrix}'\begin{bmatrix}\mathbf{u}^*{}' &1\end{bmatrix}$. And, this $\mathbf{u}^*$ is the optimal solution to the original optimization problem and satisfies $\hat Z(U^*)=Z(\mathbf{u}^*)$ and $\hat V(U^*)=V(\mathbf{u}^*)$.
\end{rem}

In the case that $L>1$, we develop a slightly weaker result:
\begin{prop}\label{prop:RankReduction}
If the optimization problem in \eqref{eq:SDP} has an optimal solution $U_1^*$, then there exists an optimal solution $U_2^*$ satisfying
\begin{equation}\label{eq:RankBound}
\frac{{\rm{rank}}(U_2^*)\left({\rm{rank}}(U_2^*)+1\right)}{2}\le K+2.
\end{equation}
\end{prop}

\
The proof of this proposition provides a simple rank reduction algorithm and we can employ this algorithm in order to search for an optimal solution of the lowest possible rank.

\section{SUBOPTIMAL SOLUTIONS VIA THE SDR AND RANDOMIZATION}\label{sec:quality}

If an optimal solution $U^*$ of one of the SDPs in \eqref{eq:InputModelSDP} and \eqref{eq:ModelInputSDP} satisfies ${\rm{rank}}(U^*)=1$ as studied in Section \ref{sec:Rank}, then the solution can be easily decomposed into an optimal solution of the original optimization problem. If $U^*$  is of low rank, but greater than 1, then this in a sense reduces the dimensionality of the search space for a control input. In fact, there are some cases where it can be proved that near-optimal feasible solutions can be generated from the relaxation and sampling schemes.

Considering  problem \eqref{eq:ModelInput} with input amplitude is constrained, i.e. with $Z_a(\mathbf{u})=Z_{i\infty}(\mathbf{u})$, and the weighted average model discrimination as an objective function, i.e. $V_2(\mathbf u)$, then the results of \cite[Sec 4.1]{Nesterov:1998} can be directly applied to prove that
\begin{equation}\nonumber
\frac{2}{\pi}\hat V_2(U^*)\le V_2(\mathbf{u}^*)\le \hat V_2(U^*).
\end{equation}
Furthermore, a simple randomization procedure with rounding will achieve that accuracy in expectation. We omit details here due to space, but it is the same as in the proof of the main result in \cite{Nesterov:1998}, and also applied in \cite{Manchester:2010}.

Moving to more general problems, consider the optimization problem in \eqref{eq:InputModel} with the choice of \eqref{eq:Discrimination} and its stochastic version
\begin{equation}\nonumber
\begin{split}
&\min_{Q,q} E\left[Z(Q\xi+q)^2\right]\\
&\hspace{1mm}{\rm{s.t.}}\ E\left[\frac{1}{\overline\gamma_m^2}\left|\overline G_m\left(Q\xi+q\right)+\overline\eta_m\right|^2\right]\ge 1,\ m=1,\cdots,M
\end{split}
\end{equation}
with a random vector $\xi\in\mathbb{R}^T$ that has a standard normal distribution. Then, we can rewrite this stochastic optimization problem as
\begin{equation}\label{eq:InputPowerModelStochastic}
\begin{split}
&\min_{Q,q} Z_{SDP}\left(\begin{bmatrix}QQ'+qq' &q\\ q' &1\end{bmatrix}\right)^2\\
&\hspace{1mm}{\rm{s.t.}}\ \frac{1}{\overline\gamma_m^2}{\rm{Tr}}\left(\begin{bmatrix}\overline G_m{}'\\ \overline\eta_m{}'\end{bmatrix}\begin{bmatrix}\overline G_m &\overline\eta_m\end{bmatrix}\begin{bmatrix}QQ'+qq' &q\\ q' &1\end{bmatrix}\right)\ge 1,\\
&\hspace{59mm}m=1,\cdots,M.
\end{split}
\end{equation}
It can be shown that this optimization problem is the same as the SDP in \eqref{eq:InputModelSDP} in the sense that $U^*=\begin{bmatrix}Q^*Q^*{}'+q^*q^*{}' &q^*\\ q^*{}' &1\end{bmatrix}$ for an optimal solution $U^*$ to the SDP in \eqref{eq:InputModelSDP} and an optimal solution $Q^*$ and $q^*$ to the SDP in \eqref{eq:InputPowerModelStochastic}. Thus, the SDP in \eqref{eq:InputModelSDP} can be viewed as a stochastic version of the original optimization problem in \eqref{eq:InputModel}. Similarly, the other SDPs can be viewed as stochastic versions of their original optimization problems in \eqref{eq:InputModel} and \eqref{eq:ModelInput}, respectively.

In light of the fact that the SDPs are stochastic versions of the original optimization problems, randomization approaches (e.g. \cite{Luo:2007} and \cite{Manchester:2012}) are natural choices in order to extract, from the solutions of the SDPs, a feasible solution to the original optimization problems. We employ the following algorithm for each optimization problem.
\begin{alg}\label{alg:fea}
Given an SDP from \eqref{eq:InputModelSDP} and \eqref{eq:ModelInputSDP}, denote its optimal solution by $U^*=\begin{bmatrix}Q^*Q^*{}'+q^*q^*{}' &q^*\\ q^*{}' &1\end{bmatrix}$.

Step $1$ : Generate a realization $\xi\in\mathbb{R}^T$ of a standard normal distribution.

Step $2$ :  Search for a constant $a^*$ such that $a^*Q^*\xi+q^*$ (i) is a feasible solution to the original optimization problem corresponding to the SDP and (ii) produces a better optimal value for the original optimization problem than $aQ^*\xi+q^*$ with any other constant $a$. If such a constant does not exist, go to Step $1$.

Step $3$ : Update $\mathbf{\hat u}=a^*Q^*\xi+q^*$ if this vector $a^*Q^*\xi+q^*$ produces the best objective value so far through this algorithm. If the number of the generations of $\xi$ is less than a certain positive number, then go to Step $1$. Otherwise, terminate the algorithm.
\end{alg}

Even though the vectors that are generated in Step $1$ suggest a solution to the original optimization problem in the sense of average, each vector may not satisfy the constraints of the original optimization problem. Thus, the vectors are scaled by $a^*$ in Step $2$ in order to meet the constraints and, at the same time, to find a better solution than $Q^*\xi+q^*$. This can be viewed as a line search.

It is possible that a constant $a^*$ in Step $2$ does not exist for some $\xi$. For example, consider the optimization problem in \eqref{eq:InputModel} with a choice of \eqref{eq:Discrimination} and its corresponding SDP in \eqref{eq:InputModelSDP} for Algorithm \ref{alg:fea}. If, for a vector $\xi$ generated in Step $1$, there exists an $m\in\{1,\cdots,M\}$ such that $\overline G_mQ^*\xi=0$ and $\frac{1}{\bar\gamma_m^2}\left\|\overline G_mq^*+\overline\eta_m\right\|^2<1$, then the scaling scheme does not produce a feasible solution and, hence, Algorithm \ref{alg:fea} returns back to Step $1$ and generate another vector. However, since $\overline G_m$'s and $Q^*$ are not zero matrices and a random vector $\xi$ has a continuous PDF, we have $P\left[\overline G_mQ^*\xi\right]=0$, which means that there exists $a^*$ in Step $2$, with probability $1$. Thus, every iteration of Algorithm \ref{alg:fea} produces, with probability $1$, a feasible solution to the original optimization problems.

When Algorithm \ref{alg:fea} is performed on an optimization problem whose optimal value is approximately known, we can modify Algorithm \ref{alg:fea} to be terminated in Step $3$ if the current vector $a^*Q^*\xi+q^*$ produces an objective which is sufficiently accurate. In the next section, we construct such an approximation of the optimal value.

\subsection{Quality of the suboptimal solutions}

Although the SDPs and Algorithm \ref{alg:fea} can provide fairly good solutions with high probability, the optimal values of the original optimization problems are unknown. Instead, in this section, we attain, for some optimization problems in \eqref{eq:InputModel} and \eqref{eq:ModelInput}, regions where the optimal values reside in. These regions can provide some ideas about the accuracy of the computations via SDR.

We consider the optimization problem in \eqref{eq:InputModel} and its SDP in \eqref{eq:InputModelSDP} in Lemma \ref{lem:InputModel} below. This lemma is an extension from Theorem 1 in \cite{Luo:2007} and, thus, the proof of the lemma follows, in general, the proof of the theorem.

\begin{lem}\label{lem:InputModel}
Let $\mathbf{u}^*$ and $U^*=\begin{bmatrix}Q^*Q^*{}'+q^*q^*{}' &q^*\\ q^*{}' &1\end{bmatrix}$ be the optimal solutions of the optimization problems in \eqref{eq:InputModel} with choices of $V_\infty(\mathbf{u})$ in \eqref{eq:Discrimination} and $Z(\mathbf{u})$ in \eqref{eq:InputPower} or \eqref{eq:InputAmplitude} and its corresponding SDP in \eqref{eq:InputModelSDP}, respectively. Then, we have
\begin{equation}\nonumber
\hat Z(U^*)\le Z(\mathbf{u}^*)\le\sqrt{\frac{27}{\pi\rho}}(M+1)E\left[Z(Q^*\xi)^2\right]^\frac{1}{2}+Z(q^*)^\frac{1}{2}
\end{equation}
where $\xi\in\mathbb{R}^T$ is a random vector with a standard normal distribution and
\begin{equation}\nonumber
\rho=\min_{m\in\{1,\cdots,M\}}\sum_{i=1}^T\lambda_{m,i}^2
\end{equation}
where $\lambda_{m,i}$, $i=1,\cdots,T$, are the singular values of $\frac{1}{\overline\gamma_m}\overline G_mQ^*$.
\end{lem}

\

Note that, with the choice of $Z(\mathbf{u})=Z_{i2}(\mathbf{u})=\frac{1}{\bar u^2}\left|\mathbf{u}\right|^2$, we have $E\left[Z(Q^*\xi)^2\right]^\frac{1}{2}=\frac{1}{\bar u}{\rm{Tr}}(Q^*Q^*{}')^\frac{1}{2}$ and $Z(q^*)=\frac{1}{\bar u^2}\left|q^*\right|^2$. Further, if $\overline\eta_m=\mathbf{0}$ $\forall m\in\{1,\cdots,M\}$, then Proposition \ref{prop:rho} below shows that $q^*=\mathbf{0}$ and $\rho\ge 1$, from which it follows that $E\left[Z(Q^*\xi)^2\right]^\frac{1}{2}=\hat Z(U^*)$ and $Z(q^*)=0$ and we have
\begin{equation}\nonumber
\hat Z(U^*)^2\le Z(\mathbf{u}^*)^2\le\frac{27M^2}{\pi}\hat Z(U^*)^2
\end{equation}
where $M+1$ is replaced with $M$ since there is no additional constraint from homogenization (see Theorem $1$ in \cite{Luo:2007}).

The upper bound in Lemma \ref{lem:InputModel} depends on $\rho$, the least sum of the squares of the singular values of  $\frac{1}{\overline\gamma_m}\overline G_mQ^*$. A larger value for $\rho$ is preferable for a tighter bound and the following proposition provides its property.

\begin{prop}\label{prop:rho}
Let $\mathbf{u}^*$ and $U^*=\begin{bmatrix}Q^*Q^*{}'+q^*q^*{}' &q^*\\ q^*{}' &1\end{bmatrix}$ be optimal solutions to the optimization problem in \eqref{eq:InputModel} with choices of $V(\mathbf{u})$ in \eqref{eq:Discrimination} and $Z(\mathbf{u})$ in \eqref{eq:InputPower} or \eqref{eq:InputAmplitude} and the SDP in \eqref{eq:InputModelSDP}, respectively. The constant $\rho$ in Lemma \ref{lem:InputModel} satisfies
\begin{equation}\nonumber
\rho\ge 1-\max_{m\in\{1,\cdots,M\}}\frac{1}{\overline\gamma_m^2}\left|\overline G_mq^*+\overline\eta_m\right|^2.
\end{equation}
Furthermore, in the case that $\overline\eta_m=\mathbf{0}$ $\forall m\in\{1,\cdots,M\}$, either we have $q^*=\mathbf{0}$ or a matrix $\begin{bmatrix}Q^*Q^*{}' &\mathbf{0}\\ \mathbf{0}' &1\end{bmatrix}$ is also an optimal solution to the SDP in \eqref{eq:InputModelSDP}, which leads to
$\rho\ge 1.
$
\end{prop}
\ 

\section{AN EXAMPLE}\label{sec:example}

In this section, the input signal design algorithm is applied to a fault detection problem for wind turbines.

A pitch angle $y$ of a blade of a wind turbine is the angle between the rotor plane and the blade chord line and, thus, a pitch angle $y=0^\circ$ means that the blade is aligned in parallel with the rotor plane. The blade is rotated by a hydraulic system and a popular model of this actuator is a closed-loop transfer function between the pitch angle $y$ and a reference angle $u$
\begin{equation}\nonumber
\frac{\omega^2}{s^2+2\zeta\omega s+\omega^2}
\end{equation}
where $\zeta$ and $\omega$ are the damping ratio and the natural frequency, respectively. See, for example, \cite{Odgaard:2013} for the details. In a normal condition, the parameters are $\zeta_1=0.6$ and $\omega_1=11.11$.

There are two major faults that can happen in the pitch angle control and we consider only one of them that is caused by an abrupt drop of the hydraulic pressure. In this case, the parameters change to $\zeta_2=0.45$ and $\omega_2=5.73$.

In order to detect the fault, i.e. to distinguish between two models based on their input-output signal, we first discretize two models of the actuators to obtain $\mathcal{G}_1$ and $\mathcal{G}_2$ corresponding to $(\zeta_1,\omega_1)$ and $(\zeta_2,\omega_2)$, respectively. The discretization is performed with a sampling time $0.01s$ and a zero-order hold. In order to complete the model structures as in \eqref{eq:Models}, we use identity operators for both $\mathcal{H}_1$ and $\mathcal{H}_2$ and the initial conditions, the current pitch angle and pitch angular velocity, of both models have mean vectors $\bar x_1=\bar x_2=\begin{bmatrix}0.5^\circ &0^\circ/s\end{bmatrix}'$ and $Q_1=Q_2=I$ for their covariances. And, we assume that the values $\sigma_1$ and $\sigma_2$ are less than or equal to $\bar\sigma=\sqrt{2}$.

Then, we search for an input signal, for a time horizon $T=100$, using an optimization problem in \eqref{eq:ModelInput} with choices of $V_\infty(\mathbf{u})$ in \eqref{eq:Discrimination} and $Z_{i2}(\mathbf{u})$ in \eqref{eq:InputPower} combined with $\bar u=1.5$, i.e. we maximize the level of model discrimination while keeping the power of the input signal below $\bar u$. As guaranteed by Remark \ref{rem:Rank1}, its corresponding SDP in \eqref{eq:ModelInputSDP} has an optimal solution of rank $1$. Thus, this SDP is solved by the CVX followed by a rank reduction procedure. The designed input signal is shown in Fig. \ref{fig:Designed} (Top).
\begin{figure}[t]
\centering
\includegraphics[width=\columnwidth]{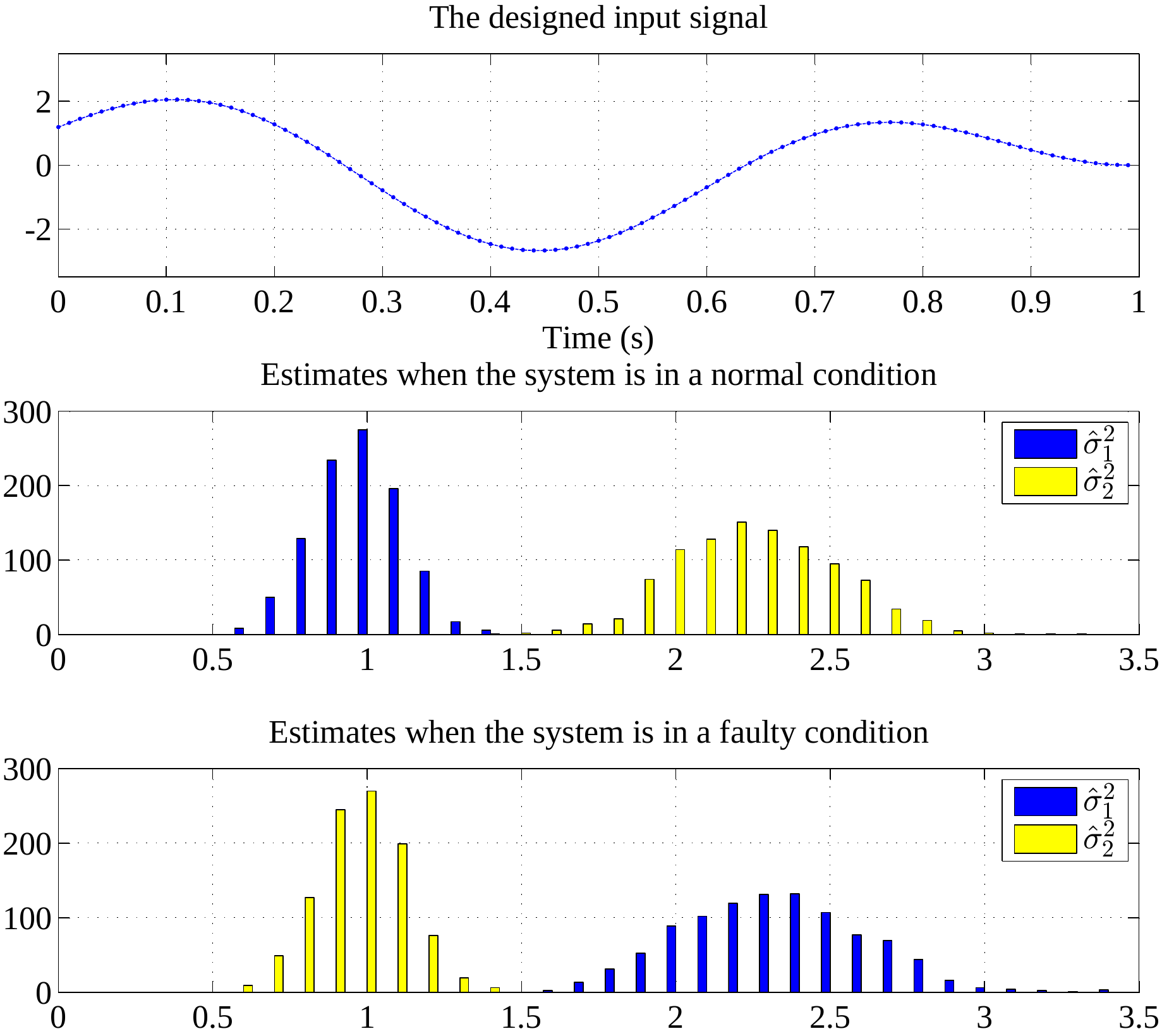}
\caption{The designed input signal (Top) and the corresponding empirical PDFs of $\hat\sigma_1^2$ and $\hat\sigma_2^2$ when the system is in the normal condition (Middle) and in the faulty condition (Bottom).}
\label{fig:Designed}
\end{figure}

First, we apply the designed input signal to the first model $\mathcal{G}_1$, which corresponds to the normal condition, with  $\sigma_1=1$ from time $0$ to $T=100$ and, then, compute estimates $\hat\sigma_1^2$ and $\hat\sigma_2^2$. This simulation is repeated $1000$ times to obtain empirical PDFs of $\hat\sigma_1^2$ and $\hat\sigma_2^2$, which are shown in Fig. \ref{fig:Designed} (Middle). As shown in the figure, the second model, which corresponds to the faulty condition, produces greater values for the estimate with high probability and, thus, the model discrimination method in Section \ref{sec:ModelDiscrimination} selects the first model, which is the correct model, with high probability. In Fig. \ref{fig:Designed} (Bottom), empirical PDFs of $\hat\sigma_1^2$ and $\hat\sigma_2^2$ are shown when the system is in the faulty condition, from which it is evident that the model discrimination method selects the correct model, which is the second model, with high probability.

For comparison, we employ a step input signal with $\bar u$ as its amplitude, shown in Fig. \ref{fig:Another} (Top), for the same simulation and obtain empirical PDFs of $\hat\sigma_1^2$ and $\hat\sigma_2^2$ in Fig. \ref{fig:Another} (Middle) and (Bottom) when the system is in the normal condition and the faulty condition, respectively.
\begin{figure}[t]
\centering
\includegraphics[width=\columnwidth]{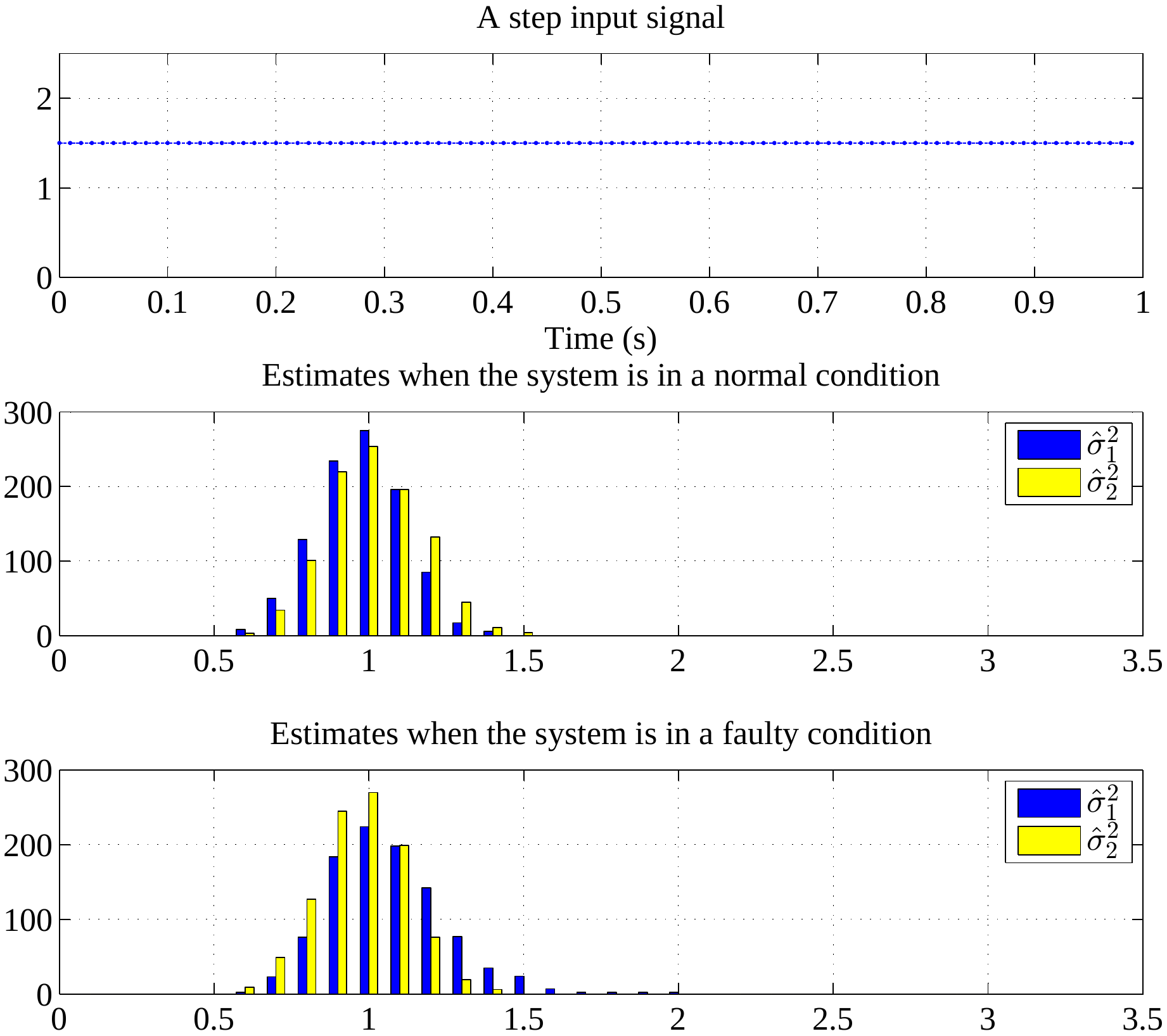}
\caption{A step input signal (Top) and the corresponding empirical PDFs of $\hat\sigma_1^2$ and $\hat\sigma_2^2$ when the system is in the normal condition (Middle) and in the faulty condition (Bottom).}
\label{fig:Another}
\end{figure}
This input signal is common in practice but, as can be seen from the figures, model discrimination is impossible.

\section{CONCLUSION}\label{sec:conclusion}

In this paper, we give a procedure for the design of probing input signals for model discrimination and fault detection over finite time intervals. The design method uses a likelihood-based model selection criteria, which we obtain from a modification of PEM to accommodate probabilistic structures of initial conditions of models. From this, we obtain conditions on input signals that guarantee that the hypothesis testing distinguishes models from each other with a given level of confidence.

From this general setting, several specific optimization problems are constructed from different constraints on the input, and different criteria for model discrimination. These optimization problems are nonconvex, and difficult to solve in general, so we suggest a solution procedure based on semidefinite relaxation and sampling.

The quality of this relaxation scheme is assessed based on some known results from duality of nonconvex quadratic programming and randomization algorithms. The utility of the method is assessed with an example of fault detection in a wind turbine.

An interesting future application will be to combine this design scheme with some preexisting control laws to form a so-called dual control, that optimizes both model discrimination ability and some other control objectives. This may be useful for switching adaptive control. It is straightforward to include in our proposed method linear and quadratic costs, so a natural controller to combine with would be MPC.



\section*{APPENDIX}

\subsection{Proof of Theorem \ref{prop:ID}}

Since $p_{n_1}$ is a Gaussian random vector with a zero mean vector and a covariance $\sigma_{n_1}^2I$, we have, from the definition of the value $\chi_{1-\alpha,T_{n_1}-1}$,
\begin{equation}\nonumber
P\left[\frac{|p_{n_1}|^2}{\sigma_{n_1}^2}\le\chi_{1-\alpha,T_{n_1}-1}^2\right]=1-\alpha,
\end{equation}
from which, together with \eqref{eq:SigmaHatn} and the fact that $|\tilde p_{n_1}|=\left|\begin{bmatrix}\Psi_{n_1}Q_{n_1} &H_{n_1}\end{bmatrix}^+\begin{bmatrix}\Psi_{n_1}Q_{n_1} &H_{n_1}\end{bmatrix}p_{n_1}\right|\le |p_{n_1}|$ and $\sigma_{n_1}\le\bar\sigma$, it follows that
\begin{equation}\label{eq:SigmaHatn1}
P\left[\frac{T_{n_1}\hat\sigma_{n_1}^2}{\bar\sigma^2}\le\chi_{1-\alpha,T_{n_1}-1}^2\right]\ge 1-\alpha.
\end{equation}

Pick any $n_2\in\left\{1,\cdots,N\right\}\setminus\{n_1\}$. (i) In the case that the condition in \eqref{eq:MuTildeCondition1} holds, it follows, from \eqref{eq:ptilden1n2}, that
\begin{equation}\nonumber
\begin{split}
&\left|\tilde p_{n_2}\right|\ge\left|\tilde\mu_{n_2n_1}\right|-\left|\tilde\Sigma_{n_2n_1}\tilde p_{n_1}\right|\\
&>\chi_{1-\alpha,T_{n_2}-1}\bar\sigma+\left(\frac{\chi_{1-\alpha,T_{n_1}-1}\bar\sigma}{\left|\tilde p_{n_1}\right|}-1\right)\left|\tilde\Sigma_{n_2n_1}\tilde p_{n_1}\right|
\end{split}
\end{equation}
and, thus, we have $\left|\tilde p_{n_2}\right|>\chi_{1-\alpha,T_{n_2}-1}\bar\sigma$ when $\left|\tilde p_{n_1}\right|\le\chi_{1-\alpha,T_{n_1}-1}\bar\sigma$, which implies that
\begin{equation}\label{eq:Pn2n1}
P\left[\frac{T_{n_2}\hat\sigma_{n_2}^2}{\bar\sigma^2}>\chi_{1-\alpha,T_{n_2}-1}^2\left|\frac{T_{n_1}\hat\sigma_{n_1}^2}{\bar\sigma^2}\le\chi_{1-\alpha,T_{n_1}-1}^2\right.\right]=1.
\end{equation}
(ii) In the other case that the condition in \eqref{eq:MuTildeCondition2} holds, we obtain, using \eqref{eq:ptilden1n2},
\begin{equation}\nonumber
\begin{split}
&\left|\tilde p_{n_1}\right|\ge\left|\tilde\mu_{n_1n_2}\right|-\left|\tilde\Sigma_{n_1n_2}\tilde p_{n_2}\right|\\
&>\chi_{1-\alpha,T_{n_1}-1}\bar\sigma+\left(\frac{\chi_{1-\alpha,T_{n_2}-1}\bar\sigma}{\left|\tilde p_{n_2}\right|}-1\right)\left|\tilde\Sigma_{n_1n_2}\tilde p_{n_2}\right|
\end{split}
\end{equation}
and, thus, we have $\left|\tilde p_{n_2}\right|>\chi_{1-\alpha,T_{n_2}-1}\bar\sigma$ if $\left|\tilde p_{n_1}\right|\le\chi_{1-\alpha,T_{n_1}-1}\bar\sigma$. This also leads to \eqref{eq:Pn2n1}, which implies that either of conditions \eqref{eq:MuTildeCondition1} and \eqref{eq:MuTildeCondition2} guarantees \eqref{eq:Pn2n1}.

Since $n_2$ is arbitrarily selected, we have
\begin{equation}\nonumber
P\left[\frac{T_{n_2}\hat\sigma_{n_2}^2}{\bar\sigma^2}>\chi_{1-\alpha,T_{n_2}-1}^2\left|\frac{T_{n_1}\hat\sigma_{n_1}^2}{\bar\sigma^2}\le\chi_{1-\alpha,T_{n_1}-1}^2\right.\right]=1
\end{equation}
$\forall n_2\in\left\{1,\cdots,N\right\}\setminus\{n_1\}$ and, thus,
\begin{equation}\nonumber
\begin{split}
&P\left[\frac{T_{n_2}\hat\sigma_{n_2}^2}{\bar\sigma^2}>\chi_{1-\alpha,T_{n_2}-1}^2\ \ \forall n_2\in\left\{1,\cdots,N\right\}\setminus\{n_1\}\right.\\
&\hspace{4mm}\left.\left|\frac{T_{n_1}\hat\sigma_{n_1}^2}{\bar\sigma^2}\le\chi_{1-\alpha,T_{n_1}-1}^2\right.\right]=1,
\end{split}
\end{equation}
from which, together with \eqref{eq:SigmaHatn1}, it follows that
\begin{equation}\nonumber
\begin{split}
&P\left[\frac{T_{n_2}\hat\sigma_{n_2}^2}{\bar\sigma^2}>\chi_{1-\alpha,T_{n_2}-1}^2\ \ \forall n_2\in\left\{1,\cdots,N\right\}\setminus\{n_1\}\right.\\
&\hspace{5mm}\left.{\rm{and}}\ \frac{T_{n_1}\hat\sigma_{n_1}^2}{\bar\sigma^2}\le\chi_{1-\alpha,T_{n_1}-1}^2\right]\\
&=P\left[\frac{T_{n_2}\hat\sigma_{n_2}^2}{\bar\sigma^2}>\chi_{1-\alpha,T_{n_2}-1}^2\ \ \forall n_2\in\left\{1,\cdots,N\right\}\setminus\{n_1\}\right.\\
&\hspace{8mm}\left.\left|\frac{T_{n_1}\hat\sigma_{n_1}^2}{\bar\sigma^2}\le\chi_{1-\alpha,T_{n_1}-1}^2\right.\right]P\left[\frac{T_{n_1}\hat\sigma_{n_1}^2}{\bar\sigma^2}\le\chi_{1-\alpha,T_{n_1}-1}^2\right]\\
&\ge 1-\alpha,
\end{split}
\end{equation}
which completes the proof.

\subsection{Proof of Proposition \ref{prop:RankReduction}}

This proposition is proved by constructing, from $U_1^*$, an optimal solution $U_2^*$ satisfying \eqref{eq:RankBound}. In the case that $\frac{{\rm{rank}}(U_1^*)\left({\rm{rank}}(U_1^*)+1\right)}{2}\le K+2$, we have $U_2^*=U_1^*$. Thus, for the remaining of the proof, we suppose that
\begin{equation}\label{eq:RankU1Star}
\frac{{\rm{rank}}(U_1^*)\left({\rm{rank}}(U_1^*)+1\right)}{2}>K+2.
\end{equation}

Since $U_1^*\ge 0$ is a symmetric matrix, we can find a unitary matrix $W$ and a diagonal matrix $\Lambda={\rm{diag}}\{\lambda_1,\cdots,\lambda_{{\rm{rank}}(U_1^*)}\}$ with $\lambda_1\ge\cdots\ge\lambda_{{\rm{rank}}(U_1^*)}>0$ satisfying $U_1^*=W\begin{bmatrix}\Lambda &\mathbf{0}\\ \mathbf{0} &\mathbf{0}\end{bmatrix}W'$.
 Then, it is clear, from \eqref{eq:RankU1Star}, that there exists a rank$(U_1^*)\times$rank$(U_1^*)$ nonzero symmetric matrix $\breve U$ satisfying
\begin{equation}\label{eq:Null}
\begin{split}
{\rm{Tr}}(P_{\ell}\breve U_1)&=0,\ \ell=\ell^*\\
{\rm{Tr}}(R_k\breve U_1)&=0,\ k=1,\cdots,K\\
{\rm{Tr}}(e_{T+1}e_{T+1}{}'\breve U_1)&=0
\end{split}
\end{equation}
where $\ell^*=\arg\max_{\ell\in\{1,\cdots,L\}}{\rm{Tr}}(P_\ell U_1^*)$ and
\begin{equation}\nonumber
\begin{split}
\breve U_1=W\begin{bmatrix}\Lambda &\mathbf{0}\\ \mathbf{0} &I\end{bmatrix}^\frac{1}{2}\begin{bmatrix}\breve U &\mathbf{0}\\ \mathbf{0} &\mathbf{0}\end{bmatrix}\begin{bmatrix}\Lambda &\mathbf{0}\\ \mathbf{0} &I\end{bmatrix}^\frac{1}{2}W'.
\end{split}
\end{equation}
And, further, we can find a unitary matrix $\breve W$ and a diagonal matrix $\breve\Lambda={\rm{diag}}\{\breve\lambda_1,\cdots,\breve\lambda_{{\rm{rank}}(\breve U)}\}$ with $\breve\lambda_1\ge\cdots\ge\breve\lambda_{{\rm{rank}}(\breve U)}$ satisfying $\breve U=\breve W\begin{bmatrix}\breve\Lambda &\mathbf{0}\\ \mathbf{0} &\mathbf{0}\end{bmatrix}\breve W'$. Notice that $1\le$ rank$(\breve U)=$ rank$(\breve\Lambda)\le$ rank$(U_1^*)$.

 Then, for any constant $a$, we have
\begin{equation}\nonumber
\begin{split}
&U_2(a)\triangleq U_1^*+a\breve U_1\\
&=W\begin{bmatrix}\Lambda &\mathbf{0}\\ \mathbf{0} &I\end{bmatrix}^\frac{1}{2}\begin{bmatrix}I+a\breve U &\mathbf{0}\\ \mathbf{0} &\mathbf{0}\end{bmatrix}\begin{bmatrix}\Lambda &\mathbf{0}\\ \mathbf{0} &I\end{bmatrix}^\frac{1}{2}W'\\
&=W\begin{bmatrix}\Lambda^\frac{1}{2}\breve W &\mathbf{0}\\ \mathbf{0} &I\end{bmatrix}\begin{bmatrix}I+\begin{bmatrix}a\breve\Lambda &\mathbf{0}\\ \mathbf{0} &\mathbf{0}\end{bmatrix} &\mathbf{0}\\ \mathbf{0} &\mathbf{0}\end{bmatrix}\begin{bmatrix}\breve W'\Lambda^\frac{1}{2} &\mathbf{0}\\ \mathbf{0} &I\end{bmatrix}W',
\end{split}
\end{equation}
from which it follows that there exists a constant $a^*$ such that
\begin{equation}\label{eq:RankReduction}
\begin{split}
{\rm{rank}}(U_2(a^*))&={\rm{rank}}\left(I+\begin{bmatrix}a^*\breve\Lambda &\mathbf{0}\\ \mathbf{0} &\mathbf{0}\end{bmatrix}\right)\\
&<{\rm{rank}}(U_1^*)
\end{split}
\end{equation}
and $I+\begin{bmatrix}a^*\breve\Lambda &\mathbf{0}\\ \mathbf{0} &\mathbf{0}\end{bmatrix}\ge 0$, which implies that
\begin{equation}\label{eq:PositiveSemidefinite}
U_2(a^*)\ge 0.
\end{equation}
Moreover, we have, from \eqref{eq:Null}, that
\begin{equation}\nonumber
\begin{split}
{\rm{Tr}}(P_\ell U_2(a^*))&={\rm{Tr}}(P_\ell U_1^*),\ \ell=\ell^*\\
{\rm{Tr}}(R_k U_2(a^*))&={\rm{Tr}}(R_k U_1^*),\ k=1,\cdots,K\\
{\rm{Tr}}(e_{T+1}e_{T+1}{}'U_2(a^*))&={\rm{Tr}}(e_{T+1}e_{T+1}{}' U_1^*),
\end{split}
\end{equation}
from which, together with \eqref{eq:RankReduction} and \eqref{eq:PositiveSemidefinite}, it follows that $U_2(a^*)$ is also an optimal solution and its rank is less than the rank of $U_1^*$.

From $U_2(a^*)$, we repeat the same procedure to obtain another matrix of a smaller rank. We repeat this rank reduction procedure until we obtain a matrix satisfying the condition in \eqref{eq:RankBound}.

\subsection{Proof of Lemma \ref{lem:InputModel}}

Due to the relaxation of the rank constraint, we have $\hat Z(U^*)\le Z(\mathbf{u}^*)$.

For a given $\xi$, we search for an appropriate constant $a\ge 0$ such that $aQ^*\xi+q^*$ is a feasible solution to the optimization problem in \eqref{eq:InputModel}. Such a constant exists with probability $1$ and we denote the constant by $a^*$. Then, it is clear that, for any given constants $\beta_1,\beta_2>0$,
\begin{equation}\nonumber
\begin{split}
&P\left[Z(a^*Q^*\xi+q^*)\le\beta_1\beta_2E\left[Z(Q^*\xi)^2\right]^\frac{1}{2}+Z(q^*)\right]\\
&\ge P\left[a^*Z(Q^*\xi)\le\beta_1\beta_2E\left[Z(Q^*\xi)^2\right]^\frac{1}{2}\right]\\
&\ge P\left[a^*\le\beta_1\ {\rm{and}}\ Z(Q^*\xi)\le\beta_2E\left[Z(Q^*\xi)^2\right]^\frac{1}{2}\right]\\
&=1-P\left[a^*>\beta_1\ {\rm{or}}\ Z(Q^*\xi)>\beta_2E\left[Z(Q^*\xi)^2\right]^\frac{1}{2}\right]\\
&\ge 1-P\left[a^*>\beta_1\right]-P\left[Z(Q^*\xi)^2>\beta_2^2E\left[Z(Q^*\xi)^2\right]\right],
\end{split}
\end{equation}
from which, together with the Markov's inequality
\begin{equation}\nonumber
P\Big[Z(Q^*\xi)^2>\beta_2^2E\left[Z(Q^*\xi)^2\right]\Big]\le\frac{1}{\beta_2^2}
\end{equation}
and Proposition \ref{prop:alpha} below, it follows that
\begin{equation}\label{eq:lem1}
\begin{split}
&P\left[Z(a^*Q^*\xi+q^*)\le\beta_1\beta_2E\left[Z(Q^*\xi)^2\right]^\frac{1}{2}+Z(q^*)\right]\\
&\ge 1-M\max\left\{\frac{1}{\beta_1\sqrt{\rho}},\frac{2\left({\rm{rank}}(U^*)-1\right)}{\left(\pi-2\right)\beta_1^2\rho}\right\}-\frac{1}{\beta_2^2}.
\end{split}
\end{equation}

It follows, from Proposition \ref{prop:RankReductionL1} and \ref{prop:RankReduction}, that ${\rm{rank}}(U^*)\le\sqrt{2(M+1)}$ and, hence, for
\begin{equation}\nonumber
\beta_1=\frac{2(M+1)}{\sqrt{\pi\rho}}\left(1-\frac{1}{\beta_2^2}\right)^{-1}{\rm{and}}\ \beta_2=\sqrt{3},
\end{equation}
we have
\begin{equation}\nonumber
\frac{1}{\beta_1\sqrt{\rho}}\ge\frac{2\left({\rm{rank}}(U^*)-1\right)}{\left(\pi-2\right)\beta_1^2\rho},
\end{equation}
from which, together with \eqref{eq:lem1}, it follows that
\begin{equation}\nonumber
\begin{split}
&P\left[Z(a^*Q^*\xi+q^*)\le\beta_1\beta_2E\left[Z(Q^*\xi)^2\right]^\frac{1}{2}+Z(q^*)\right]\\
&\ge 1-\frac{\sqrt{\pi}}{3}-\frac{1}{3}>0.
\end{split}
\end{equation}
This implies that there exists a vector $\hat\xi\in\mathbb{R}^T$ satisfying
\begin{equation}\nonumber
\begin{split}
Z(\mathbf{u}^*)&\le Z(a^*Q^*\hat\xi+q^*)\\
&\le\sqrt{\frac{27}{\pi\rho}}(M+1)E\left[Z(Q^*\xi)^2\right]^\frac{1}{2}+Z(q^*),
\end{split}
\end{equation}
which completes the proof.

\begin{prop}\label{prop:alpha}
For any $\beta_1>0$, the random variable $a^*$ in the proof of Lemma \ref{lem:InputModel} satisfies
\begin{equation}\nonumber
P\left[a^*>\beta_1\right]\le M\max\left\{\frac{1}{\beta_1\sqrt{\rho}},\frac{2\left({\rm{rank}}(U^*)-1\right)}{\left(\pi-2\right)\beta_1^2\rho}\right\}
\end{equation}
with the constant $\rho$ in Lemma \ref{lem:InputModel}.
\end{prop}
\begin{proof}
For any $\beta_1>0$, it follows, from the definition of $a^*$ in the proof of Lemma \ref{lem:InputModel}, that
\begin{equation}\label{eq:prop1}
\begin{split}
&P\left[a^*>\beta_1\right]\\
&=P\bigg[\exists\ m\in\{1,\cdots,M\}\ {\rm{s.t.}}\bigg.\\
&\hspace{9mm}\left.\frac{1}{\overline\gamma_m}\left|\overline G_m\left(\beta_1Q^*\xi+q^*\right)+\overline\eta_m\right|<1\right]\\
&\le\sum_{m=1}^MP\left[\left|\beta_1\overline G_mQ^*\xi+\overline G_mq^*+\overline\eta_m\right|<\overline\gamma_m\right].
\end{split}
\end{equation}
Using the eigenvalue decomposition, we can find two unitary matrices $V_m$ and $W_m$ and a diagonal matrix $\Lambda_m={\rm{diag}}\{\lambda_{m,1},\cdots,\lambda_{m,T}\}$ with $\lambda_{m,1}\ge\cdots\ge\lambda_{m,T}\ge 0$ satisfying $\overline G_mQ^*=V_m'\Lambda_mW_m$, which leads to
\begin{equation}\label{eq:prop2}
\begin{split}
&P\left[\left|\beta_1\overline G_mQ^*\xi+\overline G_mq^*+\overline\eta_m\right|<1\right]\\
&=P\left[\left|\beta_1\Lambda_m\tilde\xi_m+V_m\left(\overline G_mq^*+\overline\eta_m\right)\right|<1\right]\\
&\le P\left[\left|\beta_1\Lambda_m\tilde\xi_m\right|<1\right]
\end{split}
\end{equation}
where $\tilde\xi_m=\begin{bmatrix}\tilde\xi_{m,1} &\cdots &\tilde\xi_{m,T}\end{bmatrix}'=W_m\xi$. Note that $\tilde\xi_m\in\mathbb{R}^T$ is also a random vector with the standard normal distribution due to the unitary property of $W_m$.

For a given constant $\theta\in (0,1)$, if $\lambda_{m,1}^2\ge\theta\sum_{i=1}^T\lambda_{m,i}^2$, we have
\begin{equation}\label{eq:prop3}
\begin{split}
P\left[\left|\beta_1\Lambda_m\tilde\xi_m\right|<1\right]&\le P\left[\beta_1\lambda_{m,1}|\tilde\xi_{m,1}|<1\right]\\
&\le P\left[|\tilde\xi_{m,1}|<\frac{1}{\beta_1\sqrt{\theta\sum_{i=1}^T\lambda_{m,i}^2}}\right]\\
&\le\sqrt{\frac{2}{\pi\beta_1^2\theta\sum_{i=1}^T\lambda_{m,i}^2}}
\end{split}
\end{equation}
where the last inequality comes from the fact that $\tilde\xi_{m,1}$ is a random variable with a standard normal distribution.

On the other hand, if $\lambda_{m,1}^2<\theta\sum_{i=1}^T\lambda_{m,i}^2$, we have
\begin{equation}\nonumber
\left({\rm{rank}}(U^*)-1\right)\lambda_{m,2}^2\ge\sum_{i=1}^T\lambda_{m,i}^2-\lambda_{m,1}^2>(1-\theta)\sum_{i=1}^T\lambda_{m,i}^2
\end{equation}
and, hence,
\begin{equation}\nonumber
\lambda_{m,1}^2\ge\lambda_{m,2}^2>\frac{1-\theta}{{\rm{rank}}(U^*)-1}\sum_{i=1}^T\lambda_{m,i}^2,
\end{equation}
which leads to
\begin{equation}\label{eq:prop4}
\begin{split}
&P\left[\left|\beta_1\Lambda_m\tilde\xi_m\right|<1\right]\\
&\le P\left[\beta_1\lambda_{m,1}\tilde\xi_{m,1}<1\ {\rm{and}}\ \beta_1\lambda_{m,2}\tilde\xi_{m,2}<1\right]\\
&\le P\left[\beta_1\lambda_{m,1}|\tilde\xi_{m,1}|<1\right]P\left[\beta_1\lambda_{m,2}|\tilde\xi_{m,2}|<1\right]\\
&\le\frac{2\left({\rm{rank}}(U^*)-1\right)}{\pi\beta_1^2\left(1-\theta\right)\sum_{i=1}^T\lambda_{m,i}^2}
\end{split}
\end{equation}
where the last inequality comes from the fact that $\tilde\xi_{m,1}$ and $\tilde\xi_{m,2}$ are random variables with standard normal distributions.

We pick $\theta=\frac{2}{\pi}$. Then, we have, from \eqref{eq:prop3} and \eqref{eq:prop4},
\begin{equation}\label{eq:prop5}
\begin{split}
&P\left[\left|\beta_1\Lambda_m\tilde\xi_m\right|<1\right]\\
&\le\max\left\{\frac{1}{\sqrt{\beta_1^2\sum_{i=1}^T\lambda_{m,i}^2}},\frac{2\left({\rm{rank}}(U^*)-1\right)}{\left(\pi-2\right)\beta_1^2\sum_{i=1}^T\lambda_{m,i}^2}\right\}\\
&\le\max\left\{\frac{1}{\beta_1\sqrt{\rho}},\frac{2\left({\rm{rank}}(U^*)-1\right)}{\left(\pi-2\right)\beta_1^2\rho}\right\}
\end{split}
\end{equation}
with the constant $\rho$ in Lemma \ref{lem:InputModel}. The proof is completed by combining \eqref{eq:prop1}, \eqref{eq:prop2}, and \eqref{eq:prop5}.
\end{proof}

\subsection{Proof of Proposition \ref{prop:rho}}

It is clear that
\begin{equation}\label{eq:InputModelContraint}
\begin{split}
&\frac{1}{\overline\gamma_m^2}{\rm{Tr}}\left(\begin{bmatrix}\overline G_m{}'\\ \overline\eta_m{}'\end{bmatrix}\begin{bmatrix}\overline G_m &\overline\eta_m\end{bmatrix}U^*\right)\\
&=\frac{1}{\overline\gamma_m^2}\left({\rm{Tr}}\left(\overline G_mQ^*(\overline G_mQ^*)'\right)+\left|\overline G_mq^*+\overline\eta_m\right|^2\right)\\
&=\sum_{i=1}^T\lambda_{m,i}^2+\frac{1}{\overline\gamma_m^2}\left|\overline G_mq^*+\overline\eta_m\right|^2\\
&\ge 1
\end{split}
\end{equation}
for $m=1,\cdots,M$ where the last inequality comes from the fact that $U^*$ satisfies the constraints in \eqref{eq:InputModelSDP}. This leads to
\begin{equation}\nonumber
\begin{split}
\rho&=\min_{m\in\{1,\cdots,M\}}\sum_{i=1}^T\lambda_{m,i}^2\\
&\ge 1-\max_{m\in\{1,\cdots,M\}}\frac{1}{\overline\gamma_m^2}\left|\overline G_mq^*+\overline\eta_m\right|^2.
\end{split}
\end{equation}


\bibliographystyle{IEEEtran}
\bibliography{IEEEabrv,ACC2014ModelDiscrimination}

\end{document}